\documentclass[12pt, draftclsnofoot,onecolumn]{IEEEtran}
\IEEEoverridecommandlockouts
\usepackage{textcomp}
\usepackage{cite}
\usepackage{graphicx}
\usepackage{balance}
\usepackage{color}
\usepackage{url}
\usepackage{graphicx, pstool}
\usepackage{float}
\usepackage{array}
\usepackage[english]{babel} 
\usepackage{amsmath,amssymb,amsfonts}
\usepackage{amsthm}
\usepackage{textcomp}
\usepackage{algorithmic}
\usepackage{caption}
\usepackage{subcaption}
\usepackage{breqn}
\usepackage[ruled,vlined,linesnumbered]{algorithm2e}
\usepackage{cuted}

\normalsize
\makeatletter
\renewcommand{\fnum@figure}{Fig. \thefigure}
\makeatother

\newcolumntype{P}[1]{>{\centering\arraybackslash}p{#1}}
\newtheorem{theorem}{Theorem}

\newtheorem{lemma}{Lemma}

\newtheorem{corollary}{Corollary}

\newtheorem{proposition}{Proposition}

\newtheorem{conjecture}{Conjecture}

\newtheorem{sketch}{Sketch of Proof}

\def\BibTeX{{\rm B\kern-.05em{\sc i\kern-.025em b}\kern-.08em
		T\kern-.1667em\lower.7ex\hbox{E}\kern-.125emX}}

\begin{document}
	
	\title{Joint Relay Selection and Power Control that aims to Maximize Sum-Rate in Multi-Hop Networks
	}
	
	\author{
		Shalanika Dayarathna,~\IEEEmembership{Member,~IEEE,} Rajitha Senanayake,~\IEEEmembership{Member,~IEEE,}\\  and Jamie Evans,~\IEEEmembership{Senior Member,~IEEE}
	}
	\maketitle
	
	\begin{abstract}		
		Focusing on the joint relay selection and power control problem with a view to maximizing the sum-rate, we propose a novel sub-optimal algorithm that iterates between relay selection and power control. The relay selection is performed by maximizing the minimum signal-to-interference-plus-noise-ratio (as opposed to maximizing the sum-rate) and the power control is performed using a successive convex approximation. By comparing the proposed algorithm with existing solutions via extensive simulations, we show that the proposed algorithm results in significant sum-rate gains. Finally, we analyze the two-user multi-hop network and show that optimum transmit power of at least for two transmitting nodes can be found using binary power allocation. 
	\end{abstract}
	
	
	\section{Introduction}\label{section-intro}
	Relay networks have been comprehensively studied in the literature. However, most existing works focus on single or dual-user networks and/or dual-hop relaying \cite{electronics9030443} and only limited attention has been paid to general multi-user multi-hop networks. Such general networks play an important role in applications such as ad-hoc sensor networks, unmanned aerial vehicle (UAV) communication and vehicle-to-vehicle (V2V) communication \cite{5711-8}. In this paper, we focus on general multi-user, multi-hop relay networks and fill the important gap of joint relay selection and power control that aims to maximize the achievable sum-rate.
	
	First let us consider the relay selection problem. Due to the inherent interference and the competition among multiple users the relay selection problem becomes extremely challenging when several source-destination (S-D) pairs are involved. In \cite{2809748}, a sub-optimal decentralized relay selection strategy is proposed to maximize the minimum signal-to-interference-plus-noise-ratio (SINR) of such a network, while in \cite{3177187}, an efficient algorithm based on dynamic programming is proposed to obtain the optimal relay selection. Work on the multi-user relay selection problem with focus on maximizing the overall network sum-rate in DF relay networks is limited to dual-hop relay networks. In \cite{7343553}, orthogonal channels are adopted to remove interference between S-D pairs thus simplifying the relay selection problem to an assignment problem. On the other hand, in \cite{9771581}, the interference in the second hop has been estimated to reformulate the relay selection problem as an updated assignment problem. While recent works in the area of multi-hop relay networks focus on machine learning techniques to perform relay selection \cite{e23101310}, they are limited to single-user networks.
	
	Next, let us consider the power control problem that focus on maximizing the achievable sum-rate. The power control of a relay network with amplify-and-forward (AF) relaying is considered in \cite{080485} by using the successive convex approximation known as geometric programming, which is accurate in the high SNR regime. Taking a different approach, a distributed power control strategy is proposed in \cite{2121064} for a multi-user dual-hop relay network that consists of a single relay node. Recently, the concept of SNR matching is considered as the optimum power control when multiple S-D pairs do not create interference to each other. This can be achieved due to orthogonal transmissions between different S-D pairs \cite{110906}. Transmit power control in the presence of interference is considered in \cite{2431714}, where the interference is approximated by a lower bound to reformulate the power control problem as a concave optimization problem for a given relay assignment. Taking a different approach, the optimum power control for two-user dual-hop relay networks is derived analytically in \cite{9771581}. 
	
	While existing works focus on either the relay selection problem or the power control problem in multi-user, multi-hop relay networks, as far as we are aware there have been no work that considers the joint relay selection and power control problem in this general relay network. Given the dependency between relay selection and power control, the joint optimization is important to gain the optimum sum-rate performance. In this paper, we consider a multi-hop DF relay network with multiple S-D pairs and analyze the joint relay selection and power control problem with the aim of optimizing the achievable sum-rate. The contributions of this paper are listed as follows. 
	\begin{itemize}
		\item We consider five relay selection strategies that has been proposed for maximization of minimum SNR and analyze their suitability when the objective is achievable sum-rate maximization. This contribution is presented in Section \ref{section-relay}. We show that the dynamic programming based relay selection strategy with the objective of maximizing the minimum SINR achieves better sum-rate performance compared to other sub-optimal algorithms with the sum-rate maximization objective.
		\item As the main contribution, the joint relay selection and power control problem is considered for a general multi-user, multi-hop relay network and a sub-optimal algorithm that uses the dynamic programming based max-min relay selection and the tight lower bound approximation based power control is proposed. This result is presented in Algorithm \ref{Algorithm_72}. Furthermore, the performance of our proposed sub-optimal algorithm is compared against the existing resource allocation techniques, revealing that the proposed algorithm has better achievable sum-rate performance compared to the existing techniques.
		\item Under the special case of two-user multi-hop relay networks, we prove that the optimum power allocation such that the achievable sum-rate is maximized can be found analytically. This is achieved when at least two transmitting nodes transmit with binary power allocation. This contribution is presented in Theorem \ref{theorem_7.2}.
	\end{itemize} 	
	The rest of the paper is organized as follows. In Section \ref{section-model}, we provide the system model and the optimization problem formulation for a multi-user, multi-hop DF relay network with multiple relay nodes in each hop. Next, the relay selection problem is analyzed in Section \ref{section-relay} while the proposed solution and the sub-optimal algorithm is given in Section \ref{section-algo} with numerical examples in Section \ref{section-simulation}. A special case of two-user network is analyzed with respect to power control in Section \ref{section-special} and finally, the conclusions are given in Section \ref{section-conclusion}.
	
	\section{System Model and Optimization Problem Formulation}\label{section-model}	
	We study a multi-user relay network with $N$ S-D pairs as illustrated in Fig.~ \ref{figure71}. The information transmitted by the source nodes are carried over to the corresponding destination nodes via a multi-hop relay network composed of $L$ hops where each hop consists of $M$ DF relays such that $M\!\ge \!N$. Similar to \cite{2809748,130815,3177187}, we consider that a given relay node can support only one S-D pair and in each hop a given S-D pair is supported by only one relay node. This is motivated by the benefits in terms of minimization of network power consumption, processing complexities in relay nodes and the relay synchronization requirements. Therefore, $N$ relays are selected in each hop from $M$ relays and the relay node chosen in hop $l$ for S-D pair $i$ is denoted as $r_{i,l}$. Similar to \cite{2008.107,071030,5982498,6364160}, inter-hop interference is neglected due to scheduling of transmissions. 
	\begin{figure}
		\centering
		\includegraphics[width=0.85\textwidth]{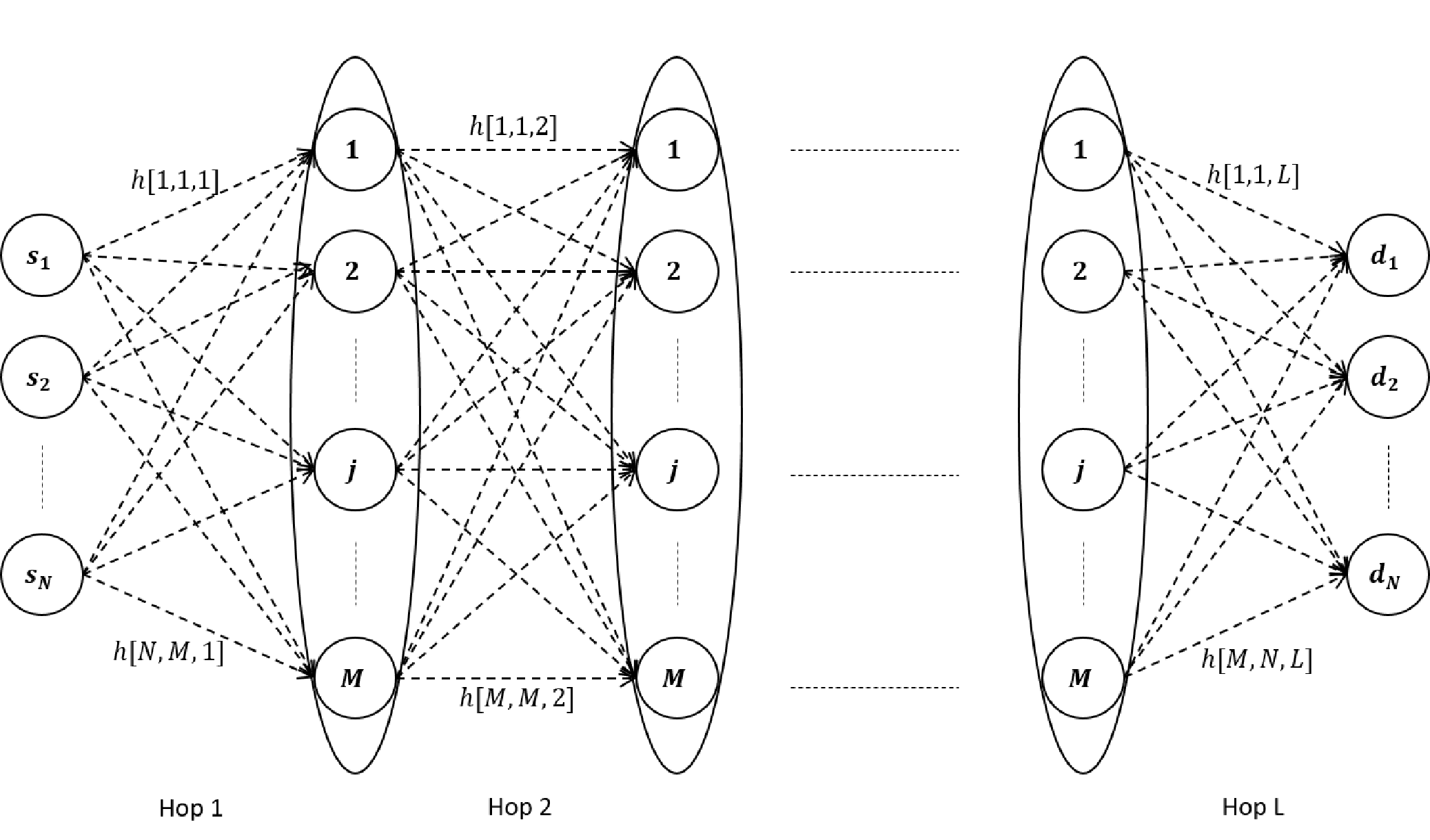}
		\centering\caption{A multi-user, multi-hop relay network}
		\label{figure71}
	\end{figure}
	Further, the channel gain between receiver $j$ and transmitter $i$ in hop $l$ is modeled as a random variable denoted by $h[i,j,l]$. When we consider the communication in hop $l$, the received signal at node $j$ can be expressed as,
		\begin{align}
	y[j,l]=\sum_{i=1}^{N}h[r_{i,l-1},j,l]\, x[r_{i,l-1},l-1]+n[j,l],
	\end{align}
	where $x[r_{i,l-1},l\!-\!1]$ is the transmitted data symbol from node $r_{i,l-1}$ in hop $l\!-\!1$, $E\{|x[r_{i,l-1},l\!-\!1]|^2\}\!=\!P[r_{i,l-1},l\!-\!1]$ with $P[r_{i,l-1},l\!-\!1]$ denoting the transmit power of node $r_{i,l-1}$ in hop $l\!-\!1$ and $n[j,l]$ is the additive white Gaussian noise (AWGN) with mean zero and variance $\sigma^2$. 
	
	In DF relay networks, the minimum SINR across all the hops for a given S-D pair determines its received end-to-end SINR. Thus, for a given relay assignment and power allocation the achievable rate for S-D pair $i$ can be written as,
	\begin{align}
	R_i = \log_2 \bigg(1+{\underset{l \in \{1,...,L\}} {\textrm{min} }}\{\gamma[i,l]\}\bigg),
	\end{align}
	where $\gamma[i,l]$ denotes the received SINR at node $r_{i,l}$ and given by,
	\begin{align}\label{snr7_1}
	\gamma[i,l] = \dfrac{P[r_{i,l-1},l-1]|h[r_{i,l-1},r_{i,l},l]|^2}{\sigma^2+\sum_{j \neq i}^N P[r_{j,l-1},l-1]|h[r_{j,l-1},r_{i,l},l]|^2},
	\end{align}
	with $r_{i,0}=r_{i,L}=i$. Next, we consider joint relay selection and power control that aims to maximize the achievable sum-rate and formulate the optimization problem as, 
	\begin{align}
	& {\underset{r_{i,l},P[r_{i,l},l] \; \forall i,l} {\textrm{max} }}\; \sum_{i=1}^{N}\log_2\bigg(1+ {\underset{l \in \{1,...,L\}} {\textrm{min} }}\;\{\gamma[i,l]\}\bigg) \nonumber\\
	&{\rm{s.t \;\; \;}}
	0 \le P[r_{i,l},l] \le P ~~ \forall l, i, 
	\nonumber\\
	& \qquad r_{i,l} \neq r_{j,l} ~~ \forall l, i \neq j, 
	\nonumber\\
	& \qquad r_{i,l} \in \{1,2,...,M\} ~~ \forall l, i,
	\label{eq_max_71}
	\end{align}
	where $P$ is the maximum transmit power for each transmission. The optimization problem in \eqref{eq_max_71} is non-convex. Due to the integer nature of $r_{i,l}$ and the existence of multiple hops combined with non-convex nature makes solving \eqref{eq_max_71} an exceptionally challenging problem for a general multi-hop relay network with multiple users. As such, we proceed with two steps to solve this optimization problem.
	
	First, we consider a given power allocation $P[r_{i,l},l] \; \forall i,l$ where $l \in \{0,\dots,L-1\}$ and write the relay selection problem as,
	\begin{align}
	& {\underset{r_{1,l},...,r_{N,l} \; \forall l} {\textrm{max} }}\; \sum_{i=1}^{N}\log_2\bigg(1+ {\underset{l \in \{1,...,L\}} {\textrm{min} }}\;\{\gamma[i,l]\}\bigg) \nonumber\\
	&{\rm{s.t \;\; \;}}
	r_{i,l} \neq r_{j,l}, ~~ \forall l, i \neq j, 
	\nonumber\\
	& \qquad r_{i,l} \in \{1,\dots,M\}, ~~ \forall l, i.
	\label{eq_max_72}
	\end{align}
	Next, we consider a given relay assignment $[r_{1,l},...,r_{N,l}], \; \forall l$ where $l \in \{1,\dots,L-1\}$ and write the power control problem as,
	\begin{align}
	& {\underset{P[r_{i,l},l] \; \forall i,l} {\textrm{max} }}\; \sum_{i=1}^{N}\log_2\bigg(1+ {\underset{l \in \{1,...,L\}} {\textrm{min} }}\;\{\gamma[i,l]\}\bigg) \nonumber\\
	&{\rm{s.t \;\; \;}}
	0 \le P[r_{i,l},l] \le P, ~~ \forall i,l.
	\label{eq_max_73}
	\end{align}
	We note that solving individual optimization problems in \eqref{eq_max_72} and \eqref{eq_max_73} separately, is still tremendously hard for a general multi-hop relay network with multiple users \cite{080485}. In the following, we first consider the optimization problems in \eqref{eq_max_72} and \eqref{eq_max_73}, separately and then propose an iterative algorithm that combines the proposed solutions in order to provide a novel joint solution.
	
	\section{Relay Selection}\label{section-relay}
	In this section we focus on the optimization problem formulated in \eqref{eq_max_72}. We note that the optimum relay assignment would involve selecting $N$ non-overlapping S-R-D paths such that each user has one distinct path and the achievable sum-rate of all users are maximized. For a relay network with $M$ relays and $L$ hops there are $M^{L-1}$ possible S-R-D path combinations for each user and $\prod_{i=0}^{N-1} (M-i)^{L-1}$ possible paths for $N$ users \cite{2809748}. 
	The achievable sum-rate depends on the combination of $N$ SINR values instead of the effective SINR of each user. As a result, for a given relay assignment in a given hop, the best path and its respective $N$ SINR values depend on the relay selection in past and future hops. Therefore, unlike the single-user network, we cannot use dynamic programming based approach (ex: Viterbi algorithm) to find the optimum relay assignment in a multi-user relay network. 
	As such, the optimal relay selection involves exponential complexity and different sub-optimal relay selection strategies with low complexity have been considered in the literature. Therefore, we consider four sub-optimal relay selection strategies used in literature 
	namely, the hop-by-hop relay selection \cite{2008.107}, ad-hoc relay selection, block-by-block relay selection \cite{071030} and sliding window based relay selection \cite{6364160}. In addition, we evaluate the achievable sum-rate obtained by the optimal relay selection when the objective is maximizing the minimum SINR \cite{3177187}.
	These strategies are detailed in the following. 
	
	\subsection{Hop-by-Hop Relay Selection}
	Under this strategy, the relay selection in each hop is performed independently such that the achievable sum-rate of the current hop is maximized when the signals are transmitted from the relays selected in the previous hop. Therefore, with this strategy, there is no sum-rate optimization involved in the last hop where the destination nodes are fixed. Since, the hop-by-hop relay selection cannot achieve full-diversity, we next consider the ad-hoc relay selection.
	
	\subsection{Ad-hoc Relay Selection}
	Under this strategy, the hop-by-hop relay selection is extended by combining the last two hops together to achieve full diversity. Therefore, while the first $L-2$ relays for each user are selected similar to the hop-by-hop relay selection, the last relay for each user is selected such that the achievable sum-rate of the last two hops is maximized.
	
	\subsection{Block-by-Block Relay Selection}
	Next, we consider the block-by-block relay selection to improve the performance further. Under this strategy, $L$ hops are divided into non-overlapping blocks of $w$ hops and the relays are selected such that the achievable sum-rate of each block is maximized. With this strategy, it is important to ensure that the block size is selected such that the last block would be greater than one. Otherwise, there will be no sum-rate optimization involved with the last hop where the destination nodes are fixed.
	
	\subsection{Sliding Window based Relay Selection}
	Taking a step further, we consider the sliding window based relay selection to remove the dependency of block size on the number of hops and to improve the performance further. Under this strategy, we consider a sliding window of $w$ hops to determine the relay selection in the first hop of the window. For example, we start by considering the first $w$ hops and find the relay selection such that the achievable sum-rate in those $w$ hops is maximized. However, we only fix the relays selected in the first hop. Next, we consider $w$ hops from the second hop to $w+1$ and fix the relays selected for the second hop. We continue this until relay selection is fixed for first $L-w$ hops. Then we consider the last $w$ hops and fix the relays selected for all of them. 
	
	\subsection{Max-Min Relay Selection}
	Finally, we consider the max-min relay selection that is used to minimize the outage probability. Under this strategy, we consider the optimal relay selection when the objective is maximizing the minimum SINR across all users. We use the dynamic programming based algorithm, which has a linear complexity with respect to $L$, proposed in \cite{3177187}. Once the relay selection is completed, we compute the achievable sum-rate according to the objective function in \eqref{eq_max_72}.
	
	We first note that when the objective is maximizing the minimum SINR, the final objective value depends only on the value of the objective function for each block of hops. As such, it is shown that the performance improves with the hop-by-hop relay selection, ad-hoc relay selection, block-by-block relay selection and sliding window based relay selection, respectively \cite{6364160}. However, when the objective is the achievable sum-rate maximization, the final effective sum-rate does not only depend on the achievable sum-rate of each block of hops. It need to be computed based on the effective minimum SINR for each user. As a result, we cannot guarantee that any of these sub-optimal relay selection strategies are always better than the others \cite{my_Thesis}.
	Therefore, in the following example, we compare the average performance of the five relay selection strategies considered in this paper.
	
	\vspace{0.25cm}
	\noindent
	\textbf{Example}: Consider a two-user ($N=2$), multi-hop relay network where the channels between nodes follow a Rayleigh distribution with zero mean and unit variance. For such a network, the gain in achievable sum-rate obtained based on different relay selection strategies compared to that of the hop-by-hop relay selection is given in Table \ref{table1}. We consider different $M$ and $L$ values with $w=2,4$ and the average received SNR of $10$ dB. In order to maintain full diversity gain, we only consider block-by-block relay selection when $L$ can be fully divided by $w$. This introduces one limitation of the block-by-block relay selection, where the block size $w$ needs to be selected depending on the number of hops $L$.
	From the table, we observe that when $L=2$, the ad-hoc relay selection, block-by-block relay selection and sliding window based relay selection have same achievable sum-rate gains. When $L=2$, all three relay selection strategies are equivalent to the optimum relay selection which considers both hops together. Thus, all three relay selection strategies result in same achievable sum-rate. Similarly, when $L=4$ both the block-by-block relay selection and sliding window based relay selection have same achievable sum-rate gains with $w=4$. 
	For any other $L$, the sliding window based relay selection has better achievable sum-rate gain for a given $M$ and $w$ compared to the ad-hoc relay selection and block-by-block relay selection. Further, we can observe that the sliding window based relay selection has much higher achievable sum-rate gain compared to the block-by-block relay selection when $w=4$. We can also observe that unlike the block-by-block relay selection, the sliding window based relay selection with $w=4$ has less sensitivity to increasing $L$ compared to $w=2$. Therefore, we can conclude that increasing $w$ provides higher performance improvements for the sliding window based relay selection compared to the block-by-block relay selection.
	\begin{table}
		\centering
		\caption{Achievable sum-rate gain percentage compared to hop-by-hop relay selection}
		\begin{tabular}{|p{1.75cm}||p{2.1cm}|p{2.1cm}|p{2.1cm}|p{2.1cm}|p{1.7cm}|p{1.7cm}|}\hline
		& \textbf{Sliding window w=2} & \textbf{Sliding window w=4} & \textbf{Block-by-block w=2} & \textbf{Block-by-block w=4} & \textbf{Ad-hoc} & \textbf{Max-Min} \\ \hline  				
		M=2, L=2 & 11.767 & - & 11.767 & - & 11.767 & 2.410 \\ \hline
		M=2, L=4 & 16.303 & 28.058 & 11.694 & 28.058 & 10.010 & 6.442 \\ \hline
		M=2, L=6 & 14.724 & 34.013 & 8.549 & - & 7.578 & 7.383 \\ \hline
		M=2, L=8 & 12.408 & 37.516 & 6.773 & 23.073 & 7.139 & 7.026 \\ \hline
		M=2, L=10 & 8.140 & 36.764 & 3.599 & - & 5.919 & 5.877 \\ \hline
		M=2, L=12 & 4.981 & 36.540 & 1.784 & 13.949 & 4.732 & 5.043 \\ \hline
		M=3, L=2 & 30.163 & - & 30.163 & - & 30.163 & 13.942 \\ \hline
		M=3, L=4 & 28.459 & 52.483 & 24.806 & 52.483 & 21.303 & 30.744 \\ \hline
		M=3, L=6 & 28.098 & 53.833 & 23.716 & - & 19.015 & 45.387 \\ \hline
		M=3, L=8 & 25.217 & 52.207 & 19.926 & 48.881 & 17.338 & 52.996 \\ \hline
		M=3, L=10 & 22.010 & 48.363 & 17.262 & - & 15.686 & 58.645 \\ \hline
		M=3, L=12 & 18.897 & 47.343 & 14.863 & 42.048 & 14.203 & 65.579 \\ \hline
		M=4, L=2 & 40.283 & - & 40.283 & - & 40.283 & 22.234 \\ \hline
		M=4, L=4 & 37.763 & 68.725 & 33.768 & 68.725 & 29.983 & 49.917 \\ \hline
		M=4, L=6 & 33.862 & 63.154 & 30.075 & - & 25.896 & 59.840 \\ \hline
		M=4, L=8 & 32.566 & 62.072 & 26.411 & 59.602 & 23.585 & 68.946 \\ \hline
		M=4, L=10 & 29.260 & 59.554 & 23.866 & - & 22.135 & 75.536 \\ \hline
		M=4, L=12 & 26.786 & 57.770 & 21.395 & 51.908 & 20.963 & 81.644 \\ \hline
		\end{tabular}
		\label{table1}
	\end{table}
 
 	When $M=2$ and $w=2$, we can also observe that the achievable sum-rate gain of the block-by-block relay selection is worse than that of the ad-hoc relay selection with increasing $L$. However, with larger $M$ block-by-block relay selection is slightly better than the ad-hoc relay selection. Similarly, when $M=2$, the max-min relay selection has the lowest achievable sum-rate gain for any given $L$. However, with increasing $M$, it outperforms all other relay selection strategies. We can also observe that for larger $M$ and $L$ values, the simple max-min relay selection provides better achievable sum-rate gains compared to the sliding window based relay selection even with $w=4$. 
 	In addition, we can also observe that with increasing $M$, the achievable sum-rate gain of all four strategies increases for a given $L$ and $w$. This can be explained by the improved diversity introduced by increasing $M$. On the other hand, with increasing $L$, the achievable sum-rate gain of the ad-hoc relay selection, block-by-block relay selection and sliding window based relay selection decreases for a given $M$ and $w$ where as that of the max-min relay selection increases for a given $M$. This can be explained by the fact that the max-min relay selection considers the channel gains of all $L$ hops while the other three strategies are unaware of the future channel gains when making the relay selection decision. Therefore, the performance of other relay selection strategies deteriorate with increasing $L$. As such, even though the objective function is different, the consideration of all hops improves the achievable sum-rate obtained with the max-min relay selection.	

	From the above example, we realized that the simple max-min relay selection strategy, which has a linear complexity with respect to $L$, provides better achievable sum-rate performance for larger relay networks where $M>N$ and $L>2$. As the optimal relay selection can be found via exhaustive search for smaller networks, in this paper we focus on the relay selection of larger multi-user, multi-hop relay networks. As such, we propose the use of max-min relay selection strategy to solve the relay selection problem in \eqref{eq_max_72}.

	\section{Joint Relay Selection and Power Control} \label{section-algo}	
	In section \ref{section-relay}, we analyzed the relay selection problem. Therefore, in this section we first focus on the optimization problem formulated in \eqref{eq_max_73} Next, we consider the power control problem for a given relay assignment which is non-convex in relation to $P[r_{i,l},l]$ \cite{080485}  and present an iterative power control algorithm. First, we use the successive convex approximation known as the tight lower bound approximation proposed in \cite{ext-thesis1,918099} and approximate its objective function as,
	\begin{align} \label{eq_max_12_1}
	\sum_{i=1}^{N}\log_2\bigg(1+{\underset{l \in \{1,...,L\}} {\textrm{min} }}\;\{\gamma[i,l]\}\bigg) \ge \dfrac{1}{\log(2)}\sum_{i=1}^N a_i \log\bigg({\underset{l \in \{1,...,L\}} {\textrm{min} }}\;\{\gamma[i,l]\}\bigg) + b_i,
	\end{align}
	that is tight at a chosen value $\mathbf{\bar{z}} = [\bar{z_1},...,\bar{z_N}]$ when the constants $a_i$ and $b_i$ are chosen as,
	\begin{align*}
	a_i = \dfrac{\bar{z}_i}{1+\bar{z}_i}, \qquad b_i = \log(1+\bar{z}_i) - \dfrac{\bar{z}_i}{1+\bar{z}_i}\log(\bar{z}_i).
	\end{align*}
	with $\bar{z}_i$ denoting the received end-to-end SINR of S-D pair~$i$, computed using the solution achieved via the previous iteration or the initial solution. Using \eqref{eq_max_12_1} we can re-write the achievable sum-rate optimization problem given in \eqref{eq_max_73} as,
	\begin{align}\label{eq_max_746}
	&{\underset{P[r_{i,l},l] \; \forall i,l}{\textrm{max} }}\; 
	\sum_{i=1}^N a_i \log\bigg({\underset{l \in \{1,...,L\}} {\textrm{min} }}\;\{\gamma[i,l]\}\bigg) + b_i \nonumber\\
	&{\rm{s.t \;\; \; }}
	0 \le P[r_{i,l},l] \le P, \forall i, l, \textrm{ where } l\in \{0,...,L-1\}.
	\end{align} 
	Using the variable transformations $P[r_{i,l},l] = e^{q[r_{i,l},l]}$ and  $t[i] = \log \bigg({\underset{l \in \{1,...,L\}} {\textrm{min} }}\;\{\gamma[i,l]\}\bigg)$, \eqref{eq_max_746} can be reformulated as,
	\begin{align}\label{eq_max_747}
	& {\underset{q[r_{i,l},l] \; \forall i,l} {  \textrm{max} }}\; 
	\sum_{i=1}^N a_i \,t[i] + b_i  \nonumber\\
	&{\rm{s.t \;\; \;}}  \nonumber \\
	&t[i] \le q[r_{i,l},l] + \log(|h[r_{i,l},r_{i,l+1},l+1]|^2) - \log\bigg(\sigma^2 + \sum_{j \neq i}^N e^{q[r_{j,l},l]}|h[r_{j,l},r_{i,l+1},l+1]|^2\bigg), \forall i, l,\nonumber\\
	&q[r_{i,l},l] \le \log(P), \; \forall i, l, 
	\end{align}	
	As a result of the convex nature of the log-sum-exp terms the optimization problem in \eqref{eq_max_747} is concave for a given relay assignment. Therefore, the coefficients $a_i$ and $b_i$ can be computed in each iteration using the results of the previous iteration. Then, approximated problem in \eqref{eq_max_747} can be solved via a gradient decent algorithm or using any existing convex solver. Due to the monotonically improving objective function resulted from the tight lower bound approximation, the sequence always converges \cite{ext-thesis1}. Therefore, for a given relay assignment, the approximated value of the optimum achievable sum-rate can be found by solving \eqref{eq_max_747} iteratively.
	
	\subsection*{Proposed Joint Solution}\label{section-joint}
	Next, we combine the iterative power control solution with relay selection and propose Algorithm~\ref{Algorithm_72} that aims to maximize the achievable sum-rate under the joint optimization. At the start of Algorithm~\ref{Algorithm_72}, the optimum achievable sum-rate, $R^{*}$, and all transmit powers are initialized to zero and $P$, respectively. In iteration $n$, we consider a given transmit power allocation and first solve the relay selection problem using the dynamic programming based max-min relay selection proposed in \cite{3177187} and assign the chosen relay nodes to a matrix marked by $\mathbf{X}$. $\mathbf{X}$ is assigned to the optimum relay assignment matrix, $\mathbf{X^{*}}$, if the resulting achievable sum-rate denoted by $R^{(n)}$ is higher than $R^{*}$. Therefore, after the first iteration, the relay selection is only changed if a different relay selection resulted in a higher $R^{(n)}$ with the updated transmit power allocation. 
	\begin{algorithm}
		\DontPrintSemicolon 
		\SetAlgoLined		
		$n=1, \; \mathbf{X^{*}}\leftarrow \{\}, \, R^*\leftarrow 0, \, P[r_{i,l},l]\leftarrow P, \forall i, l$    \;
		\While{true}{
			$[\mathbf{X},R^{(n)}] \leftarrow $ solution to relay selection problem using \cite{3177187} \;
			\If{$(R^{(n)}-R^*)/R^{(n)} > e_{\textrm{th}}$}{
				$r_{i,l} \leftarrow$ $X(l,i) \; \forall i, l$\;	
				$R^* \leftarrow R^{(n)}, \; \mathbf{X^{*}} \leftarrow \mathbf{X}$ \;
			}
			
			$m=1$ \;
			\While{true}{
				$\mathbf{Q}^{(m)} \leftarrow $ solution to problem \eqref{eq_max_747} \;
				\If{$|\;\mathbf{Q}^{(m)}-\mathbf{Q}^{(m-1)}\;|/|\;\mathbf{Q}^{(m)}\;| < e_{\textrm{th}}$}{
					break \;
				}
				m $\leftarrow$ m+1 \;
			}
			$P[r_{i,l},l] = e^{q[r_{i,l},l]}, \forall i, l$\;
			$R^{(n)} \leftarrow $ sum-rate for $r_{i,l}$ and $P[r_{i,l},l], \; \forall i,l$\;
			
			\uIf{$(R^{(n)}-R^*)/R^{(n)} > e_{\textrm{th}}$}{
				$R^* \leftarrow R^{(n)}, \; n\leftarrow n+1$ \;
			}
			\Else{
				stop \;
			}
		}
		\caption{Proposed Joint Solution}
		\label{Algorithm_72}
	\end{algorithm} 	
	Then we continue to solve the power control problem iteratively. In the $m^{\textrm{th}}$ iteration, the concave optimization problem in \eqref{eq_max_747} is solved and the solution is assigned to $L \times N$ matrix $\mathbf{Q}^{(m)}= \{q[r_{1,l},l],...,q[r_{N,l},l]\}, \forall l.$ Then a user defined threshold, $e_{\textrm{th}}$, is used to compare the calculated error. The objective function monotonically improves under the tight lower bound approximation and always converges \cite{ext-thesis1}. Therefore, within the inner loop, the achievable sum-rate increases at each iteration and $P[r_{i,l},l],\forall i,l$ is only changed at iteration $n$ if a different transmit power allocation provides a higher $R^{(n)}$ under the new relay assignment. Thus, the achievable sum-rate is monotonically improved in the $n^{\textrm{th}}$ iteration of the outer loop until it converges to a solution. Due to the non-convex nature of the optimization problem \eqref{eq_max_71}, we note there might exist multiple local peak points. As a result, the objective function would be converged to one of the local solutions in the end. Since, we cannot guarantee the convergence to the global optimum solution, the proposed algorithm is considered to be sub-optimal. In the implementation, the algorithm is considered to be converged to a local solution when the relative gap between $R^{(n)}$ and $R^{*}$ is less than $e_{\textrm{th}}$.
	
	\section{Special Case of Two-User Network} \label{section-special}
	In this section, we consider the special case of two-user networks  and prove that for at least two transmitting nodes the binary power allocation is optimum. We start by presenting Lemma \ref{lemma_7.2}.
	\begin{lemma}\label{lemma_7.2}
		An optimum power allocation that maximizes the achievable sum-rate in a multi-hop relay network with two-users can be found such that the resulting SINRs for each user in all the hopes are equal.
	\end{lemma} 
	\begin{proof} Please refer to Appendix \ref{app7:2}. \end{proof}
	
	\noindent Next, we construct the following theorem on the basis of Lemma \ref{lemma_7.2}.	
	\begin{theorem}\label{theorem_7.2}
		Optimum power allocation for at least two transmitting nodes in a two-user multi-hop DF relay network can be obtained using binary power allocation. SINR matching\footnote{Under SINR matching, the transmit power of each transmitting node in a given S-R-D path is selected such that the SINR of each receiving node in that path are equal.} can be used to obtain the optimum transmit power allocation of the rest of the nodes.
	\end{theorem} 
	
	\begin{proof}	
		\noindent According to Lemma \ref{lemma_7.2}, SINRs for each user in all the hopes are equal when the achievable sum-rate of the network is maximized.
		Therefore, the optimization problem in \eqref{eq_max_73} can be re-expressed as,
		\begin{align}
		& {\underset{P[r_{1,l},l],P[r_{2,l},l] \, \forall l} {\textrm{max} }} \bigg(1+\gamma[1,1]\bigg)\bigg(1+\gamma[2,1]\bigg) \nonumber\\
		&{\rm{s.t \;\; \;}}
		\gamma[1,1] = \gamma[1,l], \gamma[2,1] = \gamma[2,l], ~~ \forall l,
		\nonumber\\
		&  \qquad 0 \le P[r_{1,l},l],P[r_{2,l},l] \le P, ~~ \forall l, 
		\label{eq_max_744}
		\end{align}
		where $\gamma[i,l]$ is a function of $r_{i,l}$ and $ P[r_{i,l},l]$ as given in \eqref{snr7_1} with $i\in\{1,2\}$. 
		Due to double differentiability of the objective function and the equality constraints relative to $P[r_{1,l},l]$ and $P[r_{2,l},l] \, \forall l$, \eqref{eq_max_744} can be re-written as an unconstrained optimization problem
		\begin{align}
		& {\underset{\lambda^{(l)}_1,\lambda^{(l)}_2 \, \forall l} {\textrm{max} }} \bigg[{\underset{P[r_{1,l},l],P[r_{2,l},l],\, \forall l} {\textrm{min} }}\; -1 - \gamma[1,1]\gamma[2,1]  \nonumber\\ 
		& \hspace{30pt} + \bigg(\sum_{l=1}^{L-1}\lambda^{(l)}_1 - 1\bigg)\gamma[1,1] + \bigg(\sum_{l=1}^{L-1}\lambda^{(l)}_2 - 1\bigg)\gamma[2,1] \nonumber\\
		& \hspace{50pt} - \sum_{l=1}^{L-1} \biggl( \lambda^{(l)}_2 \gamma[2,1+1] - \lambda^{(l)}_1 \gamma[1,l+1]\biggr) \bigg] \nonumber\\
		&{\rm{s.t \;\; \; }} 0 \le P[r_{1,l},l],P[r_{2,l},l] \le P, ~~ \forall l, 
		\label{eq_max_745}
		\end{align}
		using the Lagrangian dual where $\lambda^{(l)}_1,\lambda^{(l)}_2 \, \forall l$ are the Lagrangian multipliers. Note that the objective function of \eqref{eq_max_745} need to be minimized to achieve the maximum achievable sum-rate. This is because in \eqref{eq_max_745}, we minimize the negative achievable sum-rate. 
		Hereafter, we use the variable $f$ to represent the objective function of \eqref{eq_max_745} and note that $f$ is a variable of $P[r_{1,l},l]$ and $P[r_{2,l},l] ~~ \forall l$. Since the Lagrangian multipliers correspond to the equality constraints, $\lambda^{(l)}_1,\lambda^{(l)}_2 \, \forall l$ can have any real value at the optimum solution of \eqref{eq_max_745}. 
		Whilst not given here due to page limitations, we can show that $f$ is not convex with respect to at least two variables. This can be done by either showing that the first derivative cannot be zero, which indicates that $f$ is either a increasing or decreasing function. In case that the first derivative can be zero, it can be shown that the second derivative is either negative or zero which indicates that $f$ is either a concave function or a increasing/decreasing function. Therefore, $f$ is minimized at the corner points implying that the maximum achievable sum-rate is obtained at zero or maximum transmit power for at least two of the variables out of $P[r_{1,l},l]$ and $P[r_{2,l},l] ~~ \forall l$. Therefore, we can conclude that at least for two transmitting nodes either zero power or maximum transmit power is optimum irrespective of the value of the Lagrangian multipliers. Finally, the values of other $2(L-1)$ transmit powers can be obtained by solving the $2(L-1)$ equality constraints in \eqref{eq_max_745} which link all $2L$ transmit powers. This concludes the proof of Theorem~\ref{theorem_7.2}. 
	\end{proof}	
	
	Thus, for a two user network, the optimum solution can be analytically obtained for the power control problem in \eqref{eq_max_73} by considering that the optimum power allocation for at least two transmitting nodes is zero or maximum power. 

	\section{Numerical and Simulation Results}\label{section-simulation}
	In this section, we present simulation results to illustrate the performance of our proposed joint relay selection and power control solution. 
		
	We note that the joint relay selection and power control problem has not been considered in the literature for a multi-user, multi-hop relay network when the objective is maximization of the achievable sum-rate. Therefore, the performance of our proposed algorithm is compared against two commonly used reference techniques namely, the greedy relay selection and the random relay selection. Under the greedy relay selection, each user selects the best path from the set of available relay nodes, using the dynamic programming based optimal relay selection proposed in \cite{5982498}, according to a priority order. Under the random relay selection, each user randomly selects a relay path without any conflict. Under both these reference techniques, we consider power control via SINR matching. The channel gain between any two nodes are computed based on both slow and fast fading. Slow fading in terms of path loss is considered by setting the distance of all $L$ hops to $2$ km with equal distance between hops and a path loss exponent $3.6$. Rayleigh fading distribution with zero mean and unit variance is considered in terms of fast fading. For all the simulation examples, the user defined threshold, $e_{\textrm{th}}$, is fixed to $10^{-3}$. AWGN noise variance, $\sigma^2\!=\!kTB$, where $k$ is the Boltzmann’s constant, $T\!=\!290$ K is the ambient temperature and $B\!=\!200$ kHz is the equivalent noise bandwidth.
	
	\begin{figure}
		\centering\includegraphics[width=0.75\textwidth]{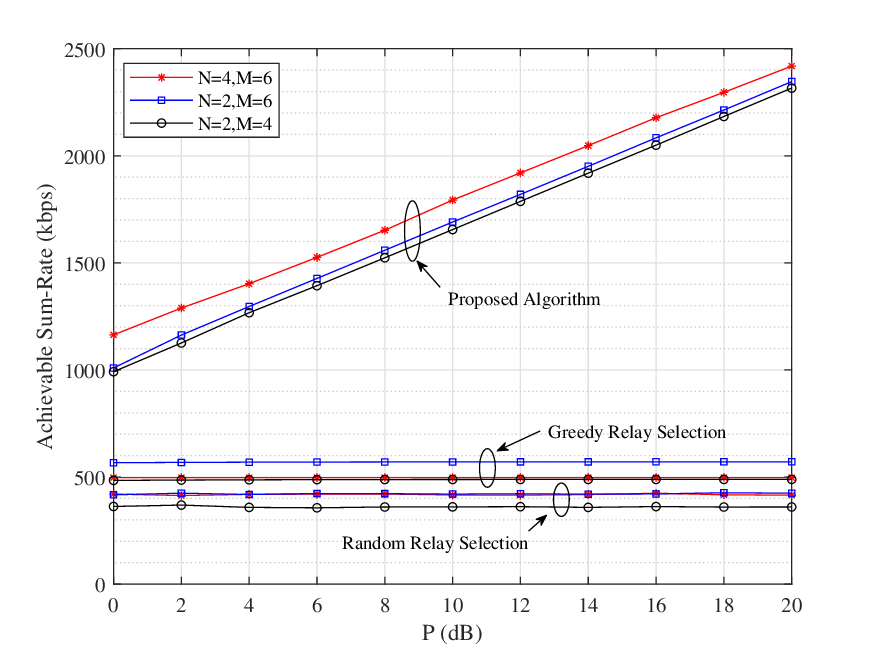}
		\captionsetup{justification=centering}
		\caption{Achievable sum-rate versus $P$ with $M=6, L=6$ dB}
		\label{figure74}
	\end{figure}
	Fig. \ref{figure74} plots the behavior of the achievable sum-rate against $P$ with $N\!=\!2,4, M\!=\!4,6$ and $L\!=\!6$, when joint optimization proposed in Algorithm \ref{Algorithm_72} is employed. From the plot, it can be observed that the resulting sum-rate from the proposed solution increases slightly with $M$ while the increment with both $P$ and $N$ is significant. When relay selection is performed by considering the interference, we can increase the overall network sum-rate with $P$ by controlling the interference while improving the received SNR. As $N$ increases, the number of summation terms in \eqref{eq_max_71} increases. Even though the individual rate of each user decreases due to the extra interference, by controlling the interference, we can increase the overall network sum-rate. As $M$ increases, the number of available relay combinations increases thus improving the probability of a user selecting a relay path with larger gain and lower interference increases. However, as the greedy relay selection does not consider interference when performing relay selection, the achievable sum-rate only depends on the number of available relay combinations. As such, the sum-rate of the greedy relay selection remains constant with $P$ and increases with $M\!-\!N$. On the other hand, the random relay selection remains constant with respect to both $P$ and $N$ while increases with $M$. As a result, the proposed algorithm improves the sum-rate of multi-user relay networks in the presence of interference, in comparison to existing relay selection solutions.

 		Fig. \ref{figure73} plots the behavior of the achievable sum-rate against $L$ with $N=2, M=6$ and $P=10$ dB, when when joint optimization proposed in Algorithm \ref{Algorithm_72} is employed. From the plot, it can be observed that the sum-rate resulted from Algorithm \ref{Algorithm_72} increases with $L$ where as that of the two reference techniques slightly decreases with $L$. For a given hop, the distance between transmitting and receiving nodes decreases with increasing $L$. This increases the sum-rate due to the reduction in the path loss between two nodes. However, as the two reference techniques make their relay selection decision purely based on the SNR, there is a high probability of selecting a relay path with high interference. As the reduction in path loss increases the interference as well, the sum-rate slightly decreases with $L$ for the greedy relay selection and random relay selection.
 	\begin{figure}
 		\centering\includegraphics[width=0.75\textwidth]{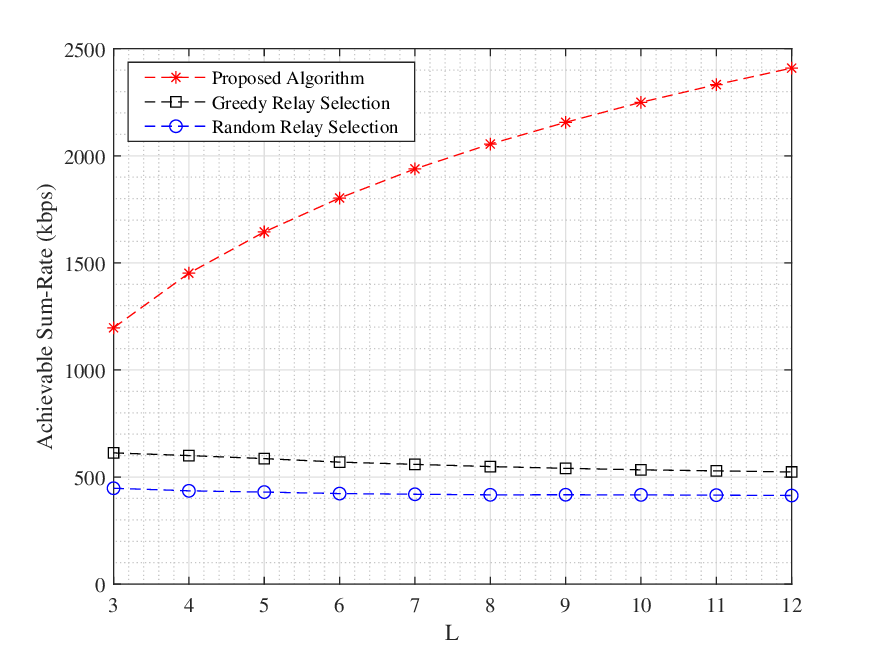}
 		\captionsetup{justification=centering}
 		\caption{Achievable sum-rate versus $L$ with $N=2, M=6, P=10$ dB}
 		\label{figure73}
 	\end{figure}
 
	Fig. \ref{figure75} plots the average computation time taken by Algorithm \ref{Algorithm_72} and the reference techniques versus $L$ when $N=2, M=6$ and $P=10$ dB. From the plot, it can be observed that the computation time of our proposed algorithm increases with $L$. We also observe that compared to the computation time of our proposed algorithm, that of the two reference techniques are significantly smaller. As such, when we consider the sum-rate performance and the complexity, a clear trade-off can be observed. As both relay selection strategies used in the greedy relay selection and Algorithm \ref{Algorithm_72} have linear complexity relative to $L$, the difference associated with the computation time is due to the iterative approach considered in Algorithm~\ref{Algorithm_72}. Therefore, we next analyze the complexity of the proposed algorithm in terms of the number of iterations for convergence.
	\begin{figure}
		\centering\includegraphics[width=0.75\textwidth]{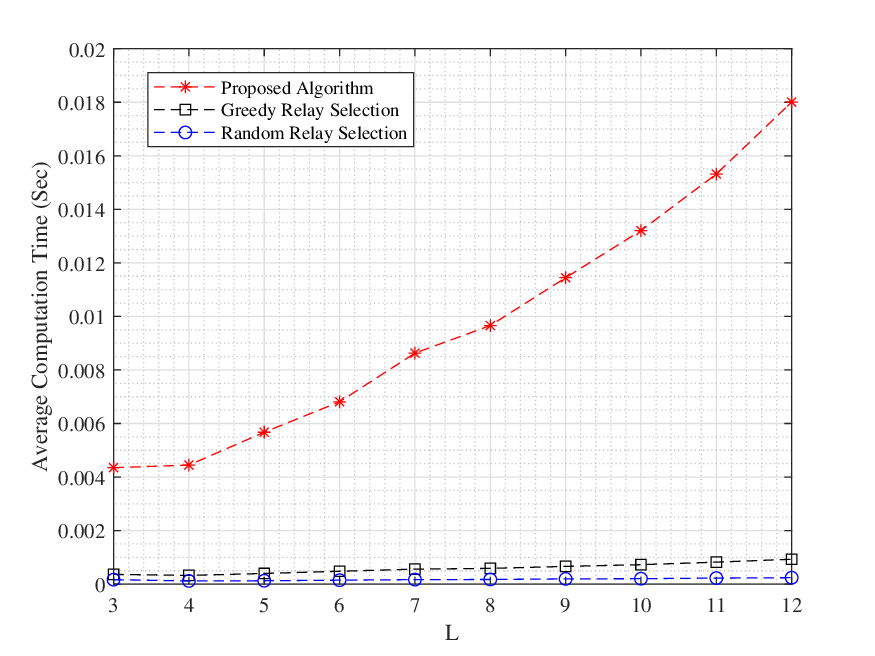}
		\captionsetup{justification=centering}
		\caption{Average computation time versus $L$ with $N=2, M=6, P=10$ dB}
		\label{figure75}
	\end{figure}

	Fig. \ref{figure76} plots the total number of iterations including both inner and outer loops of Algorithm \ref{Algorithm_72} versus $L$ when $N=2, M=6$ and $P=10$ dB. From the figure, we observe that as $L$ increases, the number of iterations increases as well. Since, the increment in the number of iterations is linear, we can conclude that the proposed algorithm has linear complexity with respect to the total number of iterations.
	\begin{figure}
		\centering\includegraphics[width=0.75\textwidth]{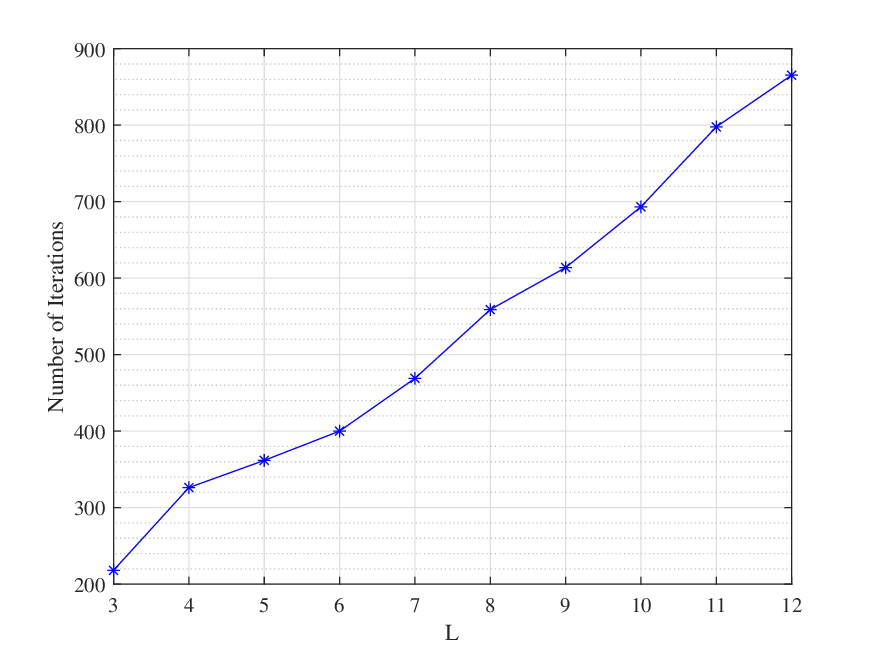}
		\captionsetup{justification=centering}
		\caption{Average number of iterations versus $L$ with $N=2, M=6, P=10$ dB}
		\label{figure76}
	\end{figure}  		
	
	\section{Conclusion} \label{section-conclusion}
	We considered the achievable sum-rate optimization problem in a general multi-user, multi-hop relay network with multiple relay nodes in each hop. First, we investigated the suitability of five sub-optimal relay selection strategies that have been considered for single-user multi-hop relay networks, namely the hop-by-hop relay selection, ad-hoc relay selection, block-by-block relay selection, sliding window based relay selection and max-min relay selection. It is shown that the dynamic programming based max-min relay selection with the objective of maximizing the minimum SINR results in higher achievable sum-rate gain compared to other sub-optimal relay selection strategies with the objective of maximizing the achievable sum-rate. Next, we combined the max-min relay selection and the tight lower bound approximation based power control to present a novel iterative algorithm. Our proposed algorithm performs joint relay assignment and power control in a such a way that the achievable sum-rate is maximized. Further, we proved that for the special case of two-user networks, binary power allocation is optimum for at least two transmitting nodes. Transmit power of other nodes can be obtained by considering that the received SINR of each user is equal over all the hops.
	
	The interference management technique that treats interference as noise is close to optimal when the interference is sufficiently weak. In addition, that reduces the complexity involved with the receivers compared to successive interference cancellation. Therefore, in this work, we considered that single user decoding is performed at each
	receiver with interference treated as noise. However, with the discussion of the research community moving towards successive interference cancellation, a desirable extension would be to consider joint relay selection and power control in the presence of successive interference cancellation.
	
	\begin{appendices}			
	\section{Proof of Lemma \ref{lemma_7.2}}\label{app7:2}
	In this section we provide the proof of Lemma \ref{lemma_7.2}. Let  $R^*$ denotes the optimum achievable sum-rate that results from the transmit power vector $[P[r_{1,l},l]^*, P[r_{2,l},l]^*] ~~ \forall l \in \{0,...,L-1\}$ with $P[r_{1,l},l]^*$ and $P[r_{2,l},l]^*$  denoting the optimum transmit powers of nodes $r_{1,l}$ and $r_{2,l}$ in hop $l$, respectively.  
	Let the resulting optimum SINRs for the $l^{\mathrm{th}}$ hop and the resulting overall minimum optimum SINRs of $s_1$ and $s_2$ be denoted by $\gamma[1,l]^{*}, \gamma[2,l]^{*},\gamma_1^{*}$ and $\gamma_2^{*}$, respectively. As a result, we can write $R^* = \log_2(1+\gamma_1^{*}) + \log_2(1+\gamma_2^{*})$.
	We start the proof by assuming that the two users do not have equal SINRs in all the hops at the same time, i.e, $\gamma[1,l]^{*} = \gamma_1^{*}$ and $\gamma[2,l]^{*} = \gamma_2^{*}$ for all $l \in \{1,...,L\}$ does not happen simultaneously. In the following, we consider the two possible scenarios resulting from the above assumption. 
	
	\subsection*{Scenario 1 - Only one user has equal SINRs in all the hops}
	Without the loss of generality, let us assume that the first user has a higher SINR in the first hop such that $\gamma[1,1]^{*} > \gamma_1^{*}$ and the second user has equal SINRs in all the hops. Next, we change the power values for $s_1$ and $s_2$ as $P[r_{1,0},0] = P[r_{1,0},0]^* - x_1$ and $P[r_{2,0},0] = P[r_{2,0},0]^* - y_1$ such that $\gamma[1,1] = \gamma_1^{*}$ and $\gamma[2,1] = \gamma_2^{*}$. 
	Based on the values of $x_1, y_1$ and considering the fact that $\gamma[1,1]^{*}>\gamma_1^{*}$, it can be shown that the new power values $P[r_{1,0},0]$ and $P[r_{2,0},0]$ falls within $0$ and $P$. 
	Therefore, we can achieve $\gamma[1,1] = \gamma_1^{*}$ and $\gamma[2,1] = \gamma_2^{*}$ for the same optimum achievable sum-rate $R^*$. Likewise, We can update the transmit power values of any hop where the first user has a higher SINR without changing $R^*$ following a similar approach. 
	
	\subsection*{Scenario 2 - None of the users have equal SINRs in all the hops}
	Under this scenario, two users either can have their corresponding higher SINRs in the same hop or in two different hops. 
	
	Let us first consider the situation where two users have their corresponding higher SINRs in the same hop. Without loss of generality, let us assume that both users have higher SINRs in the first hop such that $\gamma[1,1]^{*} > \gamma_1^{*}$ and $\gamma[2,1]^{*} > \gamma_2^{*}$. Next, we change the power values for $s_1$ and $s_2$ as $P[r_{1,0},0] = P[r_{1,0},0]^* - x_2$ and $P[r_{2,0},0] = P[r_{2,0},0]^* - y_2$ such that $\gamma[1,1] = \gamma_1^{*}$ and $\gamma[2,1] = \gamma_2^{*}$. 
	Similar to scenario 1, 
	we can show that the new power values $P[r_{1,0},0]$ and $P[r_{2,0},0]$ falls within $0$ and $P$. 
	Therefore, we can achieve $\gamma[1,1] = \gamma_1^{*}$ and $\gamma[2,1] = \gamma_2^{*}$ for the same optimum achievable sum-rate $R^*$. Again, we can update the transmit power values of any hop where both users have higher SINRs in the same hop without changing $R^*$ following a similar approach. 
	
	Let us now consider the situation where two users have their corresponding higher SINRs in two different hops. Without loss of generality, let us assume that the first user has a higher SINR in the first hop and the second user has a higher SINR in the second hop such that $\gamma[1,1]^{*} > \gamma_1^{*}$ and $\gamma[2,2]^{*} > \gamma_2^{*}$. Next, we change the power value for $s_1$ as $P[r_{1,0},0] = P[r_{1,0},0]^* - x_3$ such that $\gamma[1,1] = \gamma_1^{*}$. 
	Similar to scenario 1, we can show that $0 < P[r_{1,0},0] < P$. However, this would results in $\gamma[2,1] > \gamma_2^{*}$. Therefore, after updating the transmit power values of any hop with a higher SINR for one user, we will have a network with only one user having equal SINRs which is considered under scenario 1.
	
	The fact that under both above scenarios, we can achieve the same achievable sum-rate $R^*$ such that $\gamma[1,l] = \gamma_1^{*}$ and $\gamma[2,l] = \gamma_2^{*}$ for all $l \in \{1,...,L\}$ completes the proof of Lemma \ref{lemma_7.2}.	
				
	\end{appendices}

	\bibliography{References}
	\bibliographystyle{IEEEtran}
	
\end{document}